\documentclass[draft]{birkjour}
\usepackage[noadjust]{cite}
\usepackage{xcolor}
\RequirePackage[all]{xy}

\newtheorem{thm}{Theorem}[section]

\newtheorem{lem}[thm]{Lemma}
\newtheorem{prop}[thm]{Proposition}
\theoremstyle{definition}

\theoremstyle{remark}
\newtheorem{rem}[thm]{Remark}

\newtheorem*{ex}{Example}
\numberwithin{equation}{section}

\newcommand{\BibTeX}{B\kern-0.1emi\kern-0.017emb\kern-0.15em\TeX}
\newcommand{\XYpic}{$\mathrm{X\kern-0.3em\raisebox{-0.18em}{Y}}$-$\mathrm{pic}\,$}

\newcommand{\cl}{C \kern -0.1em \ell}  

\newcommand{\ed}{\end{document}}

\usepackage{mathtools}

\begin{document}

%
%
%
%
%
%
%
%
%

\title[Lorentz Invariance of the Multidimensional Dirac--Hestenes Equation]
 {Lorentz Invariance of the Multidimensional Dirac--Hestenes Equation}
\author[S.~Rumyantseva]{Sofia Rumyantseva}
\address{%
HSE University, 101000\\ Moscow, Russia
}
\email{srumyanceva@hse.ru}
%
\author[D.~Shirokov]{Dmitry Shirokov}
\address{%
HSE University, 101000\\ Moscow, Russia;\\
 Institute for Information Transmission Problems\\ of the Russian Academy of Sciences, 127051\\ Moscow, Russia}
\email{dshirokov@hse.ru, shirokov@iitp.ru}
\subjclass{Primary 
35Q41, 81Q05, 15A66; Secondary 70S15, 81T13}
\keywords{geometric algebra, Dirac--Hestenes equation, Lorentz invariance, Dirac equation, Clifford algebra}
\date{\today}
\begin{abstract}
This paper investigates the Lorentz invariance of the multidimensional Dirac–Hestenes equation, that is, whether the equation remains form-invariant under pseudo-orthogonal transformations of the coordinates. We examine two distinct approaches: the tensor formulation and the spinor formulation. We first present a detailed examination of the four-dimensional Dirac--Hestenes equation, comparing both transformation approaches. These results are subsequently generalized to the multidimensional case with $(1,n)$ signature. The tensor approach requires explicit invariants, while the spinor formulation naturally maintains Lorentz covariance through spin group action.
\end{abstract}
\label{page:firstblob}
\maketitle

\section{Introduction}

Geometric algebra has emerged as a powerful mathematical tool for physics, offering a unified language for describing rotations, reflections, and other transformations in arbitrary dimensions \cite{Benn,hestenes2012clifford,Hestenes2,Lasenby}. Its coordinate-free nature allows for more elegant and geometrically intuitive representations of physical laws, making it especially suitable for studying Lorentz-invariant equations such as the Dirac equation. 

The study of Lorentz transformations in the context of the Dirac equation can be approached in two distinct ways: the tensor approach and the spinor approach. The tensor approach has been previously employed in the context of the four-dimensional Dirac equation formulated in a matrix representation \cite{Sommerfeld}. In this approach, the Dirac matrices undergo transformations under the pseudo-orthogonal group ${\rm O}(1,3)$, while the wave function remains invariant. Crucially, the anticommutation relations of the Dirac matrices coincide with those of the generators of the geometric algebra $\cl_{1,3}$. This equivalence permits the reformulation of the Dirac equation within the geometric algebra formalism by substituting the matrices with the corresponding generators. In this setting, the wave function resides in the minimal left ideal of the algebra \cite{lounesto,Benn,Riesz}. Consequently, when extending the tensor approach to the multidimensional Dirac equation in geometric algebra, it follows that the generators of the algebra transform under the action of the pseudo-orthogonal group. However, this method requires the imposition of additional invariants to maintain the structural consistency of the equation.

In contrast, the spinor approach applied to the Dirac equation treats the wave function itself as a geometric object, transforming it via multiplication by an element of the spin group. This method inherently fixes the matrix representation, in other words, gamma-matrices or generators of the geometric algebra remain unchanged. Due to its conceptual elegance and mathematical efficiency, the spinor approach has become the more widely adopted framework in modern treatments of relativistic wave equations \cite{andoni2024spin,Benn}.

Within the geometric algebra formalism, the Dirac--Hestenes equation serves as a real-valued analogue to the Dirac equation \cite{Hestenes,lounesto,RumShir}. This formulation leverages multivectors in geometric algebra to represent spinors, eliminating the need for matrix representations while preserving the physical concept.

In this paper, we investigate the Lorentz invariance properties of the multidimensional Dirac--Hestenes equation through both tensor and spinor approaches. Section \ref{sec:GA} introduces the foundational geometric algebra framework and establishes the relationship between pseudo-orthogonal and spin groups. Section \ref{sec:DH} presents a rigorous proof of Lorentz invariance for the four-dimensional Dirac--Hestenes equation within the tensor approach, while also reviewing the established spinor approach. The subsequent Section \ref{secMDH} introduces the multidimensional the Dirac--Hestenes equation and explores the connection between it and the Hermitian idempotent. Section \ref{secLI} conducts an analysis of Lorentz invariance for this generalized equation under both transformation schemes. Finally, Section \ref{sec:Con} summarizes our main findings.

\section{Geometric algebra formalism}
\label{sec:GA}
One of the fundamental tools for studying problems in modern mathematical physics is geometric algebra \cite{hestenes2012clifford,Hestenes2,Lasenby}. In this paper, we consider the real geometric (Clifford) algebra $\cl_{1,n}$, generated by the elements $e^0,e^1,\ldots,e^n$, which satisfy the anticommutation relations 
\begin{equation}
    \label{generator}
    e^{\mu} e^{\nu}+e^{\nu} e^{\mu} =2\eta^{\mu\nu}e, \quad  \mu,\nu\in\{0,1,\ldots,n\},
\end{equation}
where $e$ is the identity element and $\eta$ is the diagonal matrix, where the first element is $1$ and the remaining elements on the main diagonal are $-1$:
\begin{equation}
\label{diagonalmatrix}
    \eta=(\eta^{\mu\nu})_{\mu,\nu=0}^n=\operatorname{diag}(1,-1,-1,\ldots,-1).
\end{equation}

The basis of the considering geometric algebra $\cl_{1,n}$ consists of all possible ordered products of the generators: 
\[e^{\mu_1}e^{\mu_2}\cdots e^{\mu_k}=e^{\mu_1 \mu_2\ldots \mu_k},\quad 0 \leq \mu_1<\mu_2<\cdots<\mu_k\leq n.\]

Thus, any element $U\in\cl_{1,n}$, referred to as a multivector, can be decomposed on this basis:
\[U = u e +\sum_{\mu=0}^n u_{\mu} e^{\mu} +\sum_{\mu,\nu=0,\mu< \nu}^n u_{\mu\nu} e^{\mu\nu}+\cdots+u_{01\ldots n} e^{01\ldots n},\]
where $u,u_{\mu},u_{\mu\nu},\ldots,u_{01\ldots n}$ are real scalars. Using multi-indices, this expansion can be rewritten in a more compact form:
\begin{equation}
\label{U=sum}
    U = \sum_{M} u_{M} e^M, \quad u_M\in\mathbb{R},
\end{equation}
where $M= \mu_1 \mu_2\ldots \mu_k$. The length of the multi-index $M$ is denoted by $|M|=k$, where $k=0,1,\ldots,n+1$. We denote the identity element $e$ by the element $e^M$, where $M$ is empty multi-index with $|M|=0$. If a multi-index has even length, it is called an even multi-index, otherwise it is called an odd multi-index. Note, the dimension of $\cl_{1, n}$ is $2^{n+1}$.

In Sections \ref{secMDH} and \ref{secLI}, we focus on subalgebras of $\cl_{1,n}$ generated by specific sets of basis elements. For instance, if a subalgebra is constructed by the generators with odd indices, we denote it as:  
\[\cl(e^1,e^3,e^5,e^7,\cdots)	\subset \cl_{1,n}.\]

The even subalgebra $\cl_{1,n}^{(0)}$ is defined as the linear span of all basis elements with even multi-indices:  
\[\cl_{1,n}^{(0)}=\{U\in \cl_{1,n}| U = \sum_{|M|=2k} u_M e^M\}, \quad \operatorname{dim}\cl_{1,n}^{(0)} = 2^{n}. \]
An element of the even subalgebra $\cl_{1,n}^{(0)}$  is called an even element.

Similarly, the odd subspace $\cl_{1,n}^{(1)}$ consists of all basis elements with odd multi-indices: 
\[\cl_{1,n}^{(1)}=\{U\in \cl_{1,n}| U = \sum_{|M|=2k-1} u_M e^M\}, \quad \operatorname{dim}\cl_{1,n}^{(1)} = 2^{n}.\]
An element of the odd subspace $\cl_{1,n}^{(1)}$  is called an odd element. Actually, one of the key subspaces of $Cl^{(1)}_{1,n}$ is the vector subspace $Cl^1_{1,n}$, which basis consists of all generators:
$$Cl^1_{1,n}=\{U\in \cl_{1,n}| U = \sum_{\mu=0}^{n} u_{\mu} e^{\mu}\}.$$

Let us define the group of all invertible elements of $\cl_{1,n}$:
$$Cl^{\times}_{1,n}=\{U\in \cl_{1,n}|\,\exists U^{-1}\in \cl_{1,n}:\,  UU^{-1} = U^{-1}U=e\}.$$
However, the symbol $\times$ applied to any set denotes the subset consisting of all invertible elements of that set.

In this paper, we also consider the complexified geometric algebra $\mathbb{C}\otimes\cl_{1,n}$. The decomposition of any element $U\in\mathbb{C}\otimes\cl_{1,n}$ has the same form as in \eqref{U=sum}, but the coefficients $\{u_M\}$ are complex scalars.

The wave function, that is a solution to the multidimensional Dirac--Hestenes equation, depends on Cartesian coordinates $\{x^{\mu}\}_{\mu=0}^n$ of pseudo-Euclidean space $\mathbb{R}^{1,n}$.
The metric tensor $\mathbb{R}^{1,n}$ is represented by diagonal matrix $\eta$ \eqref{diagonalmatrix}.

We also verify the invariance of the multidimensional Dirac--Hestenes equation under Lorentz transformations. It is defined as a transformation of the coordinates $\{x^{\mu}\}$ via multiplication by a matrix $P$ of an appropriate pseudo-orthogonal group:
$$x^{\mu} \mapsto \hat{x}^{\mu} = p_{\nu}^{\mu} x^{\nu},\quad P=(p_{\nu}^{\mu})_{\nu,\mu=0}^n.$$
Specifically, when $n=2d$, the matrix $P$ belongs to the indefinite pseudo-orthogonal group ${\rm O}(1,n)$:
$${\rm O}(1,n) = \{A\in \operatorname{Mat}(n+1,\mathbb{R})|A^{\top}\eta A = \eta\},$$
whereas for $n=2d-1$, it belongs to the indefinite special pseudo-orthogonal group ${\rm SO}(1,n)$:
$${\rm SO}(1,n) = \{A\in {\rm O}(1,n)|\operatorname{det}A= 1\}.$$

The specific form of the transformation depends on the parity of $n$ 
because the Lorentz group in multidimensional spacetime has different spinor representations in even and odd dimensions \cite{kaku1993quantum,borvstnik2002generate}. 
A fundamental connection exists between pseudo-orthogonal groups in $\operatorname{Mat}(n+1,\mathbb{R})$ and spin groups in $\cl_{1,n}$ \cite{Traubenberg,shirokov2013}. While pseudo-orthogonal groups describe Lorentz transformations in terms of coordinate transformations, spin groups provide an alternative perspective by acting directly on the wave function. Instead of transforming the basis elements of the geometric algebra, we multiply the wave function by an element of the spin group.

To describe a spin group, we introduce the reversion operation:
\begin{equation}
\label{reversion}
\widetilde{U}=\sum_{M} (-1)^{\frac{|M|(|M|-1)}{2}}u_{M} e^M, \quad U\in \cl_{1,n}.
\end{equation}

In the paper, we consider two classes of spin groups:
$${\rm Pin}(1,n) =\{T\in\cl^{(0)\times}_{1,n}\cup \cl^{(1)\times}_{1,n}|\,\forall x\in\cl^1_{1,n} \, T x T^{-1}\in \cl^1_{1,n};\,T^{-1}=\pm\widetilde{T}\}; $$
$${\rm Spin}(1,n) = \{T\in\cl^{(0)\times}_{1,n}|\,\forall x\in\cl^1_{1,n} \, T x T^{-1}\in \cl^1_{1,n};\,T^{-1}=\pm\widetilde{T}\}.$$
The group ${\rm Pin}(1,n)$ includes both even and odd elements of the geometric algebra, while the group ${\rm Spin}(1,n)$ consists only of even elements. These groups serve as the double covers of the corresponding pseudo-orthogonal groups.

To establish the relationship between spin and pseudo-orthogonal groups, we introduce the adjoint action of a spin group:
$${\rm ad}_T x = T x T^{-1},\quad x\in\cl^1_{1,n},$$
where $T$ belongs to the corresponding spin group.

\begin{lem}
\label{hom}
    The homomorphisms
    $${\rm ad}: {\rm Pin}(1,2d-1)\to {\rm O}(1,2d-1),$$
    $${\rm ad}:{\rm Spin}(1,2d)\to {\rm SO}(1,2d)$$
     are surjective with a kernel $\{\pm1\}$.
\end{lem}

Lemma \ref{hom} implies that the group ${\rm Pin}(1,2d-1)$ is a double cover of ${\rm O}(1,2d-1)$, while ${\rm Spin}(1,2d)$ serves as a double cover of ${\rm SO}(1,2d)$. It is worth noting that the homomorphism between ${\rm O}(1,2d)$ and ${\rm Pin}(1,2d)$ is also surjective. However, its kernel would be the set $\{\pm 1, \pm e^{01\ldots 2d}\}$. For this reason, when studying the invariance of equations under Lorentz transformations, one typically considers the subgroup ${\rm SO}(1,2d)$ in the case of odd-dimensional spaces.

\begin{rem}
    Throughout this paper, we use Einstein notation.  In this convention, repeated indices in an expression are implicitly summed over their allowable range. For instance:
    $$\sum_{\mu=0}^n A_{\mu} B^{\mu}=A_{\mu} B^{\mu},\quad \mu=0,1,\ldots,n.$$
\end{rem}

From Lemma \ref{hom}, it follows that the formula
\begin{equation}
\label{TP}
    T e^{\mu} T^{-1} = p^{\mu}_{\nu} e^{\nu}, \nu,\mu=0,1,\ldots,n
\end{equation}
establishes a one-to-two correspondence between a single matrix $P=(p^{\mu}_{\nu})$ in the pseudo-orthogonal group and the elements $\pm T$ of the corresponding spin group.

Moreover, formula \eqref{TP} can be rewritten in an equivalent form. Since $T$ belongs to a spin group, its inverse $T^{-1}$ must also be an element of the same group. Thus, we can express the transformation as:
\begin{equation}
    \label{SP}
   S^{-1} e^{\mu} S = p^{\mu}_{\nu} e^{\nu}, \nu,\mu=0,1,\ldots,n.
\end{equation}

\begin{ex}
    Let us consider the element: 
    $$S=\frac{1}{\sqrt{2}}(e-e^{12})\in\cl^{(0)}_{1,3}.$$ 
    
The inverse of $S$ is given by:
$$S^{-1} = \frac{1}{\sqrt{2}}(e+e^{12}).$$ 

To verify that $S\in {\rm Pin}(1,3)$, first, we examine the connection between the inverse element and reversion of $S$:
$$\widetilde{S}=\frac{1}{\sqrt{2}}(e+e^{12})=S^{-1}.$$

Next, we examine the transformation of the basis elements under the adjoint action:
    $$S^{-1} e^{0} S = e^0,\quad S^{-1} e^{1} S = e^2,\quad S^{-1} e^{2} S=-e^1,\quad S^{-1} e^{3} S=e^3. $$

From these transformations, we construct the corresponding matrix:
\begin{equation*}
        P= \begin{pmatrix}
             1 &0 & 0&0\\
              0 &0 & 1&0\\
              0&-1&0&0\\
              0&0&0&1
         \end{pmatrix}.
     \end{equation*}
Since $P$ satisfies the condition
$$P^{\top}\eta P = \eta,$$
it follows that $P\in {\rm O}(1,3)$.
\end{ex}

\section{Dirac--Hestenes equation in Minkowski space}
\label{sec:DH}

In this section, we examine the special case $n=3$, starting with a brief overview of the foundational concepts. The Dirac--Hestenes equation is a real analog of the Dirac equation to describe the dynamics of a particle in an electromagnetic field \cite{Hestenes,lounesto}. Unlike the Dirac equation, which employs the imaginary unit $i$, the Dirac--Hestenes equation utilizes the element $I = -e^{12}$, ensuring that all terms in the equation remain real. This real-valued nature of the equation and its solutions may offer deeper insights into the underlying geometric principles.

The mass of the particle is denoted by $m$. For simplicity, we adopt natural units where the Planck constant, the particle's charge, and the speed of light are all set to unity. The electromagnetic vector-potential, denoted as $\textbf{a}(x)$, is a vector-valued function of the position $x$ in the Minkowski space $\mathbb{R}^{1,3}$. Specifically, $\textbf{a}(x)=(a_0(x),\ldots,a_3(x)):\mathbb{R}^{1,3}\to\mathbb{R}^{4}$.

The Dirac--Hestenes equation has the form:
\begin{equation}
\label{DH3}
    e^{\mu} (\partial_{\mu} \Psi(x)  + \Psi(x) a_{\mu}(x)I) E +m \Psi(x) I=0,\quad \mu=0,1,2,3,
\end{equation}
where $x\in\mathbb{R}^{1,3}$, $\partial_{\mu}= \partial/\partial x_{\mu}$, $\{e^\mu\}_{\mu=0}^3$ are generators \eqref{generator}, $E=e^0$, $I=-e^{12}$, and $\Psi(x):\mathbb{R}^{1,3} \to \cl^{(0)}_{1,3}$.

As noted in Section \ref{sec:GA}, spin groups are related to their corresponding pseudo-orthogonal groups. Hence, Lorentz transformation in the context of the Dirac equation can be considered by using either a tensor approach or a spinor approach \cite{Sommerfeld}. 
We verify both techniques for the Dirac--Hestenes equation. 
The first approach involves transforming the geometric algebra generators $e^{\mu}$. The second considers transforming the wave function $\Psi(x)$ using elements of the spin group ${\rm Pin}(1,3)$. It is well established that the spinor approach preserve the Dirac--Hestenes equation \cite{hestenes1990zitterbewegung,lounesto}. Moreover, the elements $E$ and $I$ remain invariant under Lorentz transformation. To facilitate a clearer understanding of the multidimensional case, we first present the proofs for the case $n=3$, as the essential steps generalize naturally to higher dimensions. We begin with the tensor approach.

\begin{rem}
Throughout this article, we omit the explicit dependence on $x$ for all functions in the Dirac--Hestenes equation to maintain a more concise presentation. It is important to note that under Lorentz transformation, any function $f(x)$ transforms into $\widehat{f}(\widehat{x})$. Therefore, when we write $f$, it should be understood as $f(x)$, and when we write $\widehat{f}$, it refers to $\widehat{f}(\widehat{x})$. The precise relationship between $f$ and $\widehat{f}$ depends on the chosen approach.
\end{rem}

The Dirac--Hestenes equation with transformed variables has the following form:
\begin{equation}
\label{transfDH3}
     \widehat{e^{\mu}} (\widehat\partial_{\mu} \widehat\Psi  + \widehat\Psi \widehat a_{\mu} I) E + m \widehat\Psi I=0,\quad \mu=0,1,2,3.
\end{equation}

\begin{thm}
\label{n3tensor}
Let $\mu,\,\nu=0,1,2,3$ and the transformations be defined as follows:
\begin{align*}
    &x^{\mu} \mapsto \widehat{x^{\mu}} = p_{\nu}^{\mu} x^{\nu}, \quad P=(p_{\nu}^{\mu})\in {\rm O}(1,3), \\
   & \partial_{\mu} \mapsto \widehat\partial_{\mu} = q_{\mu}^{\nu} \partial_{\nu}, \quad Q=(q_{\mu}^{\nu})=P^{-1}, \\
   & a_{\mu} \mapsto \widehat{a}_{\mu} = q_{\mu}^{\nu} a_{\nu}, \\
    &e^{\mu} \mapsto \widehat{e^{\mu}}=p_{\nu}^{\mu}e^{\nu}, \\
    &\Psi \mapsto \widehat\Psi = \Psi.
\end{align*}

We assume that $I$ and $E$ are invariant under transformations. If $\Psi$ is a solution to Dirac--Hestenes equation \eqref{DH3}, then $\widehat\Psi$ is a solution to transformed Dirac--Hestenes equation \eqref{transfDH3}.
\end{thm}

\begin{rem}
In the tensor approach, the generators of the algebra are transformed. Consequently, the wave function $\Psi$ depends not only on the coordinates $\{x^{\mu}\}$ but also on the generators $\{e^{\mu}\}$. This implies that under Lorentz transformation, the wave function transforms as follows:
$$\widehat\Psi(\widehat{x^{\mu}},\widehat{e^{\mu}})=\Psi(x^{\mu},e^{\mu}).$$
More explicitly, substituting the transformed variables, we obtain:
$$\widehat\Psi(\widehat{x^{\mu}},\widehat{e^{\mu}}) = \Psi(q_{\nu}^{\mu}\widehat{x^{\nu}},q_{\nu}^{\mu}\widehat{e^{\nu}}).$$

\end{rem}

\begin{proof}
We denote the left-hand side of transformed Dirac--Hestenes equation~\eqref{transfDH3} by $\widehat F$. Substituting the transformed variables into $\widehat F$, we obtain:
 $$\widehat F=  p_{\nu}^{\mu}q_{\mu}^{\alpha}e^{\nu} ( \partial_{\alpha}\Psi  + \Psi a_{\alpha} I) E+ m \Psi  I,\quad \nu,\,\alpha=0,1,2,3.$$

    The matrix $Q$ is the inverse of $P$. Therefore, the connection $p_{\nu}^{\mu}q^{\alpha}_{\mu}=\delta_{\nu}^{\alpha}$ simplifies $\widehat F$ to:
     $$\widehat F= e^{\nu} ( \partial_{\nu}\Psi  + \Psi a_{\nu} I) E+ m \Psi I.$$

The resulting expression is original Dirac--Hestenes equation \eqref{DH3}, and since $\Psi$ satisfies it, we conclude that $\widehat F=0$.
   
\end{proof}

In the spinor approach, relation~\eqref{SP} is used. The spin group ${\rm Pin}(1,3)$ satisfies ${\rm ad}({\rm Pin}(1,3))={\rm O}(1,3)$, implying that for every $P\in {\rm O}(1,3)$, there exists $S\in {\rm Pin}(1,3)$. Instead of transforming the geometric algebra generators by an element $P$ from the pseudo-orthogonal group, they remain unchanged and a wave function is multiplied by an element $S$ from the spin group. However, the element $S$ can belong either to the even subalgebra $\cl^{(0)\times}_{1,3}$ or to the odd subspace $\cl^{(1)\times}_{1,3}$. The solution $\Psi$ to the Dirac--Hestenes equation is an even element, and multiplying it by an odd element produces an odd result. To address this, when $S$ belongs to the odd subspace, the transformation of $\Psi$ is modified as $S\Psi E$.

\begin{thm}
Let $\mu,\,\nu=0,1,2,3$ and the transformations be defined as:
\begin{align*}
    &x^{\mu} \mapsto \widehat{x^{\mu}} = p_{\nu}^{\mu} x^{\nu}, \quad P=(p_{\nu}^{\mu})\in {\rm O}(1,3), \\
   & \partial_{\mu} \mapsto \widehat\partial_{\mu} = q_{\mu}^{\nu} \partial_{\nu}, \quad Q=(q_{\mu}^{\nu})=P^{-1}, \\
   & a_{\mu} \mapsto \widehat{a}_{\mu} = q_{\mu}^{\nu} a_{\nu},\\
    &e^{\mu} \mapsto \widehat{e^{\mu}}=e^{\mu}, \\
    &\Psi \mapsto \widehat\Psi= \begin{cases}
        S \Psi,\quad S\in {\rm Spin}(1,3)\quad(i.e.\, P\in {\rm SO}(1,3))\\
        S \Psi E, \quad S\in {\rm Pin}(1,3)\setminus {\rm Spin}(1,3)
    \end{cases}
\end{align*}
where $P$ and $S$ are related by \eqref{SP}.

If $\Psi$ is a solution to Dirac--Hestenes equation \eqref{DH3}, then $\widehat\Psi$ is also a solution to transformed Dirac--Hestenes equation \eqref{transfDH3}.
\end{thm}

\begin{proof}
We begin by considering the case $P\in {\rm SO}(1,3)$, in other words $S\in{\rm Spin}(1,3)$. As in the proof of the previous theorem, let us denote the left-hand side of transformed Dirac--Hestenes equation \eqref{transfDH3} by $\widehat F$. Substituting the transformed variables into $\widehat F$, we obtain:  
$$\widehat F = q_{\mu}^{\alpha}e^{\mu}S\left( \partial_{\alpha} \Psi  + \Psi a_{\alpha}I\right)   E+ m S\Psi I,\quad \alpha=0,1,2,3.$$

Note that $S$ may not commute with the generators $\{e^{\mu}\}$. Therefore, we factor it out of the sums in the following way:
$$\widehat F =S\left( q_{\mu}^{\alpha}S^{-1}e^{\mu}S\left( \partial_{\alpha} \Psi  + \Psi a_{\alpha}I\right)   E+ m \Psi I\right).$$

Using relation \eqref{SP} between $S\in {\rm Spin}(1,3)$ and $P\in {\rm SO}(1,3)$, we get:
$$\widehat F =S\left( q_{\mu}^{\alpha} p_{\nu}^{\mu}e^{\nu}\left( \partial_{\alpha} \Psi  + \Psi a_{\alpha}I\right) E+ m \Psi I\right),\quad \nu=0,1,2,3.$$

Since $Q$ is the inverse of $P$, the expression simplifies to:
$$\widehat F =S\left( e^{\nu}\left( \partial_{\nu} \Psi  + \Psi a_{\nu} I\right)  E+ m \Psi I\right).$$

The term in external parentheses corresponds to original Dirac--Hestenes equation \eqref{DH3}, which equals zero because $\Psi$ is a solution to it. Then, we we conclude that $\widehat F=0$.

In the case $S\in {\rm Pin}(1,3)\setminus {\rm Spin}(1,3)$, the transformation of $\Psi(x)$ includes an additional factor $E$. Hence, $\widehat F$ has the form:
 $$\widehat F =  q_{\mu}^{\alpha}e^{\mu}S\left( \partial_{\alpha} \Psi E  + \Psi Ea_{\alpha}I\right)   E+ m S\Psi E I.$$

Note that $I$ and $E$ commute. Moving $E$ to the right within the parentheses and $S$ to the left, we rewrite $\widehat F$ as:
       \begin{equation*}
            \widehat F = S\left( q_{\mu}^{\alpha}S^{-1}e^{\mu}S\left( \partial_{\alpha} \Psi  + \Psi a_{\alpha}I\right) E + m \Psi  I\right)E.
    \end{equation*}

Once again, the term in external parentheses is original Dirac--Hestenes equation \eqref{DH3}, which vanishes. Thus, we have $\widehat F = 0$.
\end{proof}

 \begin{rem}
An alternative transformation approach exists for the four-dimensional Dirac--Hestenes equation \cite{lounesto}. The generators of the algebra remain unchanged while the variables transform as follows:
$$\hat{I}=SIS^{-1},\quad \hat{E} = S E S^{-1},\quad\hat{\Psi} = S\Psi S^{-1},$$
where $S\in {\rm Spin}_{+}(1,3)$. However, the main difficulty lies in the nontrivial form of $S \Psi$ in the multidimensional case (see the details in Proposition \ref{SPsi}), which makes it challenging to generalize this approach beyond four dimensions. Since this perspective is not the primary focus of our study, we do not extend it to the multidimensional case in the present work. Though it remains an interesting direction for future research.
 \end{rem}

Overall, the Dirac--Hestenes equation is invariant under Lorentz transformations. Both the tensor and spinor approaches can be applied as in the case of the classical Dirac equation formulated in the matrix formalism. In Section \ref{secLI}, we demonstrate that the multidimensional Dirac--Hestenes equation also admits both approaches.  

\section{Multidimensional Dirac--Hestenes equation}
\label{secMDH}

Let us recall some key aspects of the multidimensional Dirac--Hestenes equation. This equation describes the dynamics of particles in higher-dimensional spaces, extending the traditional Dirac equation to a real-valued formulation~\cite{RumShir}. As in the previous section, we denote the mass of the particle by $m$ and assume that the Planck constant, the charge of the particle, and the speed of light are equal to $1$. In the multidimensional case, we consider the pseudo--Euclidean space $\mathbb{R}^{1,n}$. The electromagnetic vector-potential $\textbf{a}(x)$ is defined as:
$$\textbf{a}(x)=(a_0(x),\ldots,a_n(x)):\mathbb{R}^{1,n}\to\mathbb{R}^{n+1}.$$

In the context of the Dirac--Hestenes equation or the Dirac equation, solutions can be classified into three different types based on the dimensionality of the space \cite{Benn,BrauerWeyl}. Notably, the form of the Dirac--Hestenes equation depends on the type of solution if $n>3$. For $n=2d-1$, a solution is a spinor. When $n=2d$, 
the solution can be either a semi-spinor or a double spinor. A semi-spinor is an element of one of the two irreducible subspaces of the full spinor space in even dimensions. In the other hand, a double spinor is a direct sum of the two semi-spinor representations.

Let us recall the form of the multidimensional Dirac--Hestenes equation for all solution types \cite{RumShir}. We begin with the case $n=2d-1$. A solution to the Dirac--Hestenes equation belongs to the even subalgebra of algebra $Q'$:
\begin{equation}
        Q' = \cl(e^0,e^1,e^2,e^3,e^5,e^7,\ldots,e^{2d-1})\subset \cl_{1,2d-1},\label{Q1}
    \end{equation}
which is constructed by the generators with odd indices, $e^0$, and $e^2$. We denote the even subalgebra of $Q'$ by $Q'^{(0)} = Q'\cap \cl^{(0)}_{1,n}$.

The multidimensional Dirac--Hestenes equation takes the form: 
\begin{equation}
    \begin{aligned}
        \label{Dir-Hest}
   & \sum_{\mu = 0,1,2,3,5,7,\ldots,2d-1}e^{\mu} (\partial_{\mu} \Psi+ \Psi a_{\mu} I) E \\
    &+\sum_{\mu = 3,5,\ldots,2d-3} (\partial_{\mu+1} \Psi + \Psi a_{\mu+1} I ) e^{\mu} E I +m\Psi I=0,
    \end{aligned}
\end{equation}
where $\Psi(x): \mathbb{R}^{1,n}\to Q'^{(0)}$, $I=-e^{12}$, and $E=e^0$. Note that the index of the first summation in \eqref{Dir-Hest} is odd or equal to $0$, $2$.

Next, we consider the case of a semi-spinor. The algebra $Q'$ is constructed in the same way as in the spinor case:
\begin{equation}
        Q' = \cl(e^0,e^1,e^2,e^3,e^5,e^7,\ldots,e^{2d-1})\subset \cl_{1,2d}.\label{Q2}
    \end{equation}
However, the form of the Dirac--Hestenes equation differs slightly: 
    \begin{equation}
\label{SemispinorDir-Hest}
\begin{aligned}
      &\sum_{\mu = 0,1,2,3,5,7,\ldots,2d-1}e^{\mu} (\partial_{\mu} \Psi+ \Psi a_{\mu} I) E\\
      &+\sum_{\mu = 3,5,\ldots,2d-3,2d-1}  (\partial_{\mu+1} \Psi + \Psi a_{\mu+1} I ) e^{\mu} E I +m\Psi I=0,
\end{aligned}
\end{equation}
where $\Psi(x): \mathbb{R}^{1,n}\to Q'^{(0)}$, $I=-e^{12}$, and $E=e^0$. The key difference lies in the second summation, which includes an additional term corresponding to the index $\mu=2d-1$.

Finally, we examine the case of a double spinor. In this scenario, the algebra $Q'$ is similar to the previous cases but also contains the generator $e^{2d}$:
 \begin{equation}
        Q' = \cl(e^0,e^1,e^2,e^3,e^5,e^7,\ldots,e^{2d-1},e^{2d})\subset \cl_{1,2d}.\label{Q3}
    \end{equation}
Therefore, the Dirac--Hestenes equation for the double spinor case incorporates the generator $e^{2d}$ into the first summation. The equation takes the form: 
\begin{equation}
\label{DoublespinorDir-Hest}
\begin{aligned}
     &\sum_{\mu = 0,1,2,3,5,7,\ldots,2d-3,2d-1,2d}e^{\mu} (\partial_{\mu} \Psi+ \Psi a_{\mu} I) E\\ &+ \sum_{\mu = 3,5,\ldots,2d-3}  (\partial_{\mu+1} \Psi + \Psi a_{\mu+1} I ) e^{\mu} E I +m\Psi I=0,
\end{aligned}
\end{equation}
where $\Psi(x): \mathbb{R}^{1,n}\to Q'^{(0)}$, $I=-e^{12}$, and $E=e^0$.

We introduce the following notations to simplify the form of the multidimensional Dirac--Hestenes equation:
$$A_{\mu} \coloneqq \partial_{\mu} \Psi+ \Psi a_{\mu} I, \quad \widehat{A}_{\mu} \coloneqq  \widehat{\partial}_{\mu} \widehat{\Psi}+ \widehat{\Psi} \widehat{a}_{\mu} I, $$
where transformed variables are denoted by variables with the hat.

Let us present the transformed equations. The transformed Dirac--Hestenes equation for the case when $n=2d-1$ has the form:
\begin{equation}
    \label{TransfDir-Hest}
    \sum_{\mu = 0,1,2,3,5,7,\ldots,2d-1}\!\!\!\!\!\! \widehat{e^{\mu}} \widehat{A}_{\mu} E+\sum_{\mu = 3,5,\ldots,2d-3}  \widehat{A}_{\mu+1}I \widehat{e^{\mu}} E  +m\widehat{\Psi} I=0.
\end{equation}
For the case of a semi-spinor, the equation is given by:
\begin{equation}
\label{TransfSemispinorDir-Hest}
    \sum_{\mu = 0,1,2,3,5,7,\ldots,2d-1}\!\!\!\!\!\! \widehat{e^{\mu}} \widehat{A}_{\mu} E+\sum_{\mu = 3,5,\ldots,2d-3,2d-1}  \widehat{A}_{\mu+1}I \widehat{e^{\mu}} E  +m\widehat{\Psi} I=0.
\end{equation}
Finally, for the case of a double spinor, the equation takes the form:
\begin{equation}
\label{TransfDoublespinorDir-Hest}
    \sum_{\mu = 0,1,2,3,5,7,\ldots,2d-3,2d-1,2d}\!\!\!\!\!\! \widehat{e^{\mu}} \widehat{A}_{\mu} E+\sum_{\mu = 3,5,\ldots,2d-3}  \widehat{A}_{\mu+1}I \widehat{e^{\mu}} E  +m\widehat{\Psi} I=0.
\end{equation}

To establish the invariance of the multidimensional Dirac–Hestenes equation under Lorentz transformations, we utilize its property of multiplication by a specific Hermitian idempotent. Initially, let us define the Hermitian conjugation operation in $\mathbb{C}\otimes\cl_{1,n}$, which depends on the signature. For the signature $(1,n)$, it has the form \cite{ndimension,Shirokov}:  
\[U^{\dagger}=e^0 \widetilde{\bar{U}} e^0=e^0\sum_{M} (-1)^{\frac{|M|(|M|-1)}{2}}\bar{u}_{M} e^Me^0, \quad U\in \mathbb{C} \otimes\cl_{1,n},\]
where $\bar{U}$ denotes complex conjugation and $\widetilde{U}$
represents reversion \eqref{reversion}.

We consider the Hermitian idempotent $t$:
\[t^2=t,\quad t^{\dagger}=t,\]
and the corresponding left ideal $L(t)$ generated by $t$:
\[L(t) = \{U\in \mathbb{C} \otimes \cl_{1,n}| U t = U\}.\]

If $L(t)$ is a minimal left ideal, meaning it contains no other left ideals except itself and $L(0)$, the corresponding idempotent $t$ is called a primitive idempotent.

All types of spinors are associated with specific Hermitian idempotents~\cite{RumShir}. 
For the case of a spinor, the Hermitian primitive idempotent is constructed as: 
\begin{equation}
\label{event}
    t = \frac{1}{2}(e+ e^0)\prod\limits_{\mu = 1}^{d-1}\frac{1}{2}(e+i e^{2\mu-1}e^{2\mu}) \in  \mathbb{C} \otimes \cl_{1,2d-1}.
\end{equation}
This idempotent projects the spinor space onto an irreducible subspace within the geometric algebra $\mathbb{C} \otimes \cl_{1,2d-1}$. For semi-spinors, which arise in $\mathbb{C} \otimes \cl_{1,2d}$, the Hermitian primitive idempotent takes a similar but extended form:
 \begin{equation}
    \label{oddtsemi}
    t = \frac{1}{2}(e+ e^0)\prod\limits_{\mu = 1}^{d}\frac{1}{2}(e+i e^{2\mu-1}e^{2\mu}) \in  \mathbb{C} \otimes \cl_{1,2d}.
    \end{equation}

Finally, for double spinors, the Hermitian idempotent is given by:
  \begin{equation}
    \label{oddtdouble}
    t = \frac{1}{2}(e+ e^0)\prod\limits_{\mu = 1}^{d-1}\frac{1}{2}(e+i e^{2\mu-1}e^{2\mu}) \in  \mathbb{C} \otimes \cl_{1,2d}.
    \end{equation}
    This idempotent, while similar in structure to the spinor case, operates within the larger geometric algebra $\mathbb{C} \otimes \cl_{1,2d}$. Note that idempotent \eqref{oddtdouble} is not primitive.

It is notable that idempotents \eqref{event}-\eqref{oddtdouble} have several properties:
\begin{equation}
\label{ItEt}
    It = i t,\quad Et = t,
\end{equation}
\begin{equation}
\label{propt}
    ie^{2\mu-1}t=e^{2\mu}t,\quad \mu = 1,2,\ldots,d',
\end{equation}
where $d'$ is the upper bound of their product in \eqref{event}-\eqref{oddtdouble}. 

\begin{ex}
Let us illustrate property \eqref{propt} with specific examples. In the case of spinors in $\mathbb{C} \otimes\cl_{1,3}$ and double spinors in $\mathbb{C} \otimes\cl_{1,4}$, we have $d'=1$. Thus, the corresponding relation simplifies to:
    $$-e^1e^2 t = it.$$

In contrast, for semi-spinors in $\mathbb{C} \otimes\cl_{1,4}$ and spinors in $\mathbb{C} \otimes\cl_{1,5}$, we have $d'=2$. Therefore, the corresponding relations are:
    $$-e^1e^2 t = -e^3e^4 t =it$$
\end{ex}

Also, we recall Lemma \ref{Yt=0} \cite{RumShir}. This lemma is essential for proving the invariance of the Dirac--Hestenes equation under Lorentz transformations.
\begin{lem}
\label{Yt=0}
    Let $Q'$ be \eqref{Q1} (\eqref{Q2} or \eqref{Q3}) and the Hermitian idempotent $t$ have form \eqref{event} (\eqref{oddtsemi} or \eqref{oddtdouble} respectively). If $Y\in~Q'^{(0)}$ and $Yt=0$, then $Y=0$.
\end{lem}

We present an approach to simplify the product of the left-hand side of the multidimensional Dirac--Hestenes equation and the Hermitian idempotent $t$, which is associated with the specific type of spinor under consideration. This simplification transforms two sums of the equation into a single sum. It is important to note that such a reduction is not achievable without multiplication by the idempotent.

\begin{prop}
\label{sum2}
    Let $F$ denote the left-hand side of Dirac--Hestenes equation \eqref{Dir-Hest} (\eqref{SemispinorDir-Hest} or \eqref{DoublespinorDir-Hest}) and $t$ have form \eqref{event} (\eqref{oddtsemi} or \eqref{oddtdouble} respectively). Then:
    \begin{equation}
    \label{Ft}
       F t = \left(\sum_{\mu=0}^n e^{\mu} A_{\mu} E +m \Psi I\right)t.
    \end{equation}
\end{prop}
\begin{proof}
We provide the proof for the case $n=2d-1$. The proofs for the other cases follow an analogous procedure. 
Shifting the index in the second sum of $F$ by $1$, the expression becomes:
        $$F =  \sum_{\mu = 0,1,2,3,5,7,\ldots,2d-1}e^{\mu} A_{\mu} E +\sum_{\mu = 4,6,\ldots,2d-2} A_{\mu} e^{\mu-1} E I +m\Psi I.$$

    Multiplying $F$ by the idempotent $t$ on the right and using properties \eqref{ItEt} and \eqref{propt}, the second sum can be rewritten as:
$$ \sum_{\mu = 4,6,\ldots,2d-2} A_{\mu} e^{\mu-1} E I t = \sum_{\mu = 4,6,\ldots,2d-2}  A_{\mu} e^{\mu} E t.$$

    Since $A_{\mu}\in Q'^{(0)}$, the generators $e^{\mu}$ (for $\mu=4,6,\ldots,2d-2$) commute with $A_{\mu}$. Combining this result with the first sum, we obtain:
    $$Ft = \left(\sum_{\mu=0}^{2d-1} e^{\mu} A_{\mu}E+m\Psi I\right)t.$$
    
This completes the proof.
\end{proof}

In the next section, we demonstrate that the multidimensional Dirac--Hestenes equations \eqref{Dir-Hest}, \eqref{SemispinorDir-Hest}, and \eqref{DoublespinorDir-Hest} are invariant under Lorentz transformations. This is achieved by applying Proposition \ref{sum2} and Lemma \ref{Yt=0} to combine the two sums into one and subsequently eliminate the idempotent $t$, thereby simplifying the analysis of Lorentz invariance.

\section{Lorentz invariance}
\label{secLI}

In this section, we investigate the Lorentz invariance of the multidimensional Dirac–Hestenes equation using two distinct approaches: the tensor approach and the spinor approach. In the tensor point of view, the generators of the geometric algebra multiply by an element of the pseudo-orthogonal group. However, this approach requires several invariant quantities to maintain consistency within the equation. In contrast, the spinor approach focuses on transforming the wave function directly through multiplication by an element of the spin group, eliminating the need for additional invariants. Due to its elegance and simplicity, the spinor approach is often preferred in the literature, as it provides a more natural representation of Lorentz symmetry~\cite{lounesto}.

We begin by analyzing the tensor approach. In this approach, we impose the condition:
\begin{equation}
    \label{rist1appr}
    -e^{2\mu-1} e^{2\mu} \widehat{t}=i\widehat{t},\quad \mu=1,2,\ldots, d',
\end{equation}
where $d'$ is the upper bound of the product in the explicit form of idempotents \eqref{event}-\eqref{oddtdouble} and $\widehat{t}$ is the idempotent after the transformation. 

It is important to note that analogous conditions hold for the original idempotent $t$, as we have shown in equation \eqref{propt}. Moreover, the analogue of equation \eqref{propt} with hatted variables also holds for the idempotent $\widehat{t}$, which follows directly from its explicit form. In this context, the elements $\{e^{2\mu-1} e^{2\mu}\}$ can be regarded as invariants under the transformation. From this invariance, condition \eqref{rist1appr} immediately follows. In the multidimensional case, these elements play a role analogous to the imaginary unit $i$, similarly to how $e^{12}$ appears in the four-dimensional Dirac--Hestenes theory.

\begin{thm}
 Let $Q'$ be \eqref{Q1} (\eqref{Q2} or \eqref{Q3}), the Hermitian idempotent $t$ have form \eqref{event} (\eqref{oddtsemi} or \eqref{oddtdouble} respectively), $\mu,\,\nu=0,1,2,\ldots,n$ and the transformations be defined as follows:
 \begin{align*}
    &x^{\mu} \mapsto \widehat{x^{\mu}} = p_{\nu}^{\mu} x^{\nu}, \quad P=(p_{\nu}^{\mu})\in \begin{cases}
        {\rm O}(1,n), &n=2d-1\\
        {\rm SO}(1,n), &n=2d
    \end{cases} \\
   & \partial_{\mu} \mapsto \widehat\partial_{\mu} = q_{\mu}^{\nu} \partial_{\nu}, \quad Q=(q_{\mu}^{\nu})=P^{-1}, \\
   & a_{\mu} \mapsto \widehat{a}_{\mu} = q_{\mu}^{\nu} a_{\nu}, \\
    &e^{\mu} \mapsto \widehat{e^{\mu}}=p_{\nu}^{\mu}e^{\nu}, \\
    &\Psi \mapsto \widehat\Psi = \Psi.
\end{align*}

   If $\Psi\in Q'^{(0)}$ is a solution to Dirac--Hestenes equation \eqref{Dir-Hest} (\eqref{SemispinorDir-Hest} or \eqref{DoublespinorDir-Hest} respectively), $I$ and $E$ are invariant under transformations, and condition \eqref{rist1appr} be hold, then $\widehat\Psi$ is also a solution to transformed Dirac--Hestenes equation \eqref{TransfDir-Hest} (\eqref{TransfSemispinorDir-Hest} or \eqref{TransfDoublespinorDir-Hest} respectively).
\end{thm}

\begin{rem}
    The transformed Dirac--Hestenes equation and $\widehat\Psi$ do not belong to the algebra $Q'^{(0)}$. However, they are elements of the algebra $\widehat Q'^{(0)}$, where, for instance, for the case $n=2d-1$:
    $$\widehat Q' = \cl(\widehat{e^0},\widehat{e^1},\widehat{e^2},\widehat{e^3},\widehat{e^5},\widehat{e^7},\ldots,\widehat{e^{2d-1}})\subset \cl_{1,2d-1}.$$
\end{rem}

\begin{proof}
Let us start by considering the case $n=2d-1$. We verify the invariance of the Dirac--Hestenes equation by substituting the transformed variables into equation \eqref{Dir-Hest} and checking that the left-hand side equals zero. 
Let us denote the left-hand side of  transformed Dirac--Hestenes equation \eqref{TransfDir-Hest} by $\widehat{F}$.
The variable $\widehat{A}_{\mu}$ satisfies the transformation:
$$\widehat{A}_{\mu} = q_{\mu}^{\alpha} A_{\alpha}.$$

Multiplying $\widehat{F}$ by the idempotent $\widehat{t}$, we apply Proposition \ref{sum2} to merge the summations:  $$\widehat{F} \widehat{t} = \left(\widehat{e^{\mu}} \widehat{A}_{\mu} E +m\widehat{\Psi} I \right)\widehat{t}.$$

Thus, substituting the transformations, we obtain:
\begin{equation}
\label{tmpTheor1}
   \widehat{F} \widehat{t} =\left( p_{\nu}^{\mu} q_{\mu}^{\alpha} e^{\nu} A_{\alpha} E  +m\Psi I \right)\widehat{t},\quad \alpha=0,1,2,\ldots,2d-1. 
\end{equation}

From the fact that $Q$ is the inverse of the matrix $P$, it follows that $p_{\nu}^{\mu} q_{\mu}^{\alpha} = \delta^{\alpha}_{\nu}$. 
Therefore, taking into account condition \eqref{rist1appr} and Proposition~\ref{sum2}, the right-hand side of \eqref{tmpTheor1} is the left-hand side of equation \eqref{Dir-Hest}, which equals zero since $\Psi$ is a solution to it:
$$\widehat{F}\widehat{t}=0.$$

Applying Lemma \ref{Yt=0} for $\widehat{F}\in \widehat{Q}'^{(0)}$, we obtain that it equals zero. 

   The proofs for the other cases follow the same steps as in the case $n=2d-1$. Specifically, we substitute the transformed variables, simplify the expression by merging two summations using the corresponding idempotent $\widehat{t}$, and apply Lemma \ref{Yt=0} to conclude the result.
\end{proof}

 In the spinor approach to the Dirac equation, the generators remain unchanged, while the solution $\psi$ transforms as $S\psi$ if $S\in {\rm Spin}(1,n)$. 
However, this transformation cannot be directly applied to solutions of the Dirac--Hestenes equation. If $\Psi \in Q'^{(0)}$ and $S\in {\rm Pin}(1,n)$, then in general, $S \Psi \notin Q'^{(0)}$, since $S$ may contain the generators $e^4,e^6,\ldots,e^{2d-2}$, meaning that $S\notin Q'$.
On the other hand, a relationship exists between solutions of the multidimensional Dirac equation and those of the multidimensional Dirac--Hestenes equation. More precisely, there is a correspondence between elements of the left ideal $L(t)$ and the algebra $ Q'^{(0)}$, given by \cite{RumShir}:
    $$\forall \psi\in L(t)\, \exists! \Psi\in Q'^{(0)}: \psi=\Psi t. $$
Moreover, from the definition of the left ideal, it follows that for any $S\in\cl_{1,n}$ and $\psi\in L(t)$, the product $S \psi$ belongs to $L(t)$. Thus, an analogous result to the spinor approach for the multidimensional Dirac equation should also hold for the multidimensional Dirac--Hestenes equation. We present Proposition~\ref{SPsi} on the transformation of $\Psi$.

Actually, any element $S\in\cl^{(0)}_{1,n}$ can be decomposed by separating elements belonging to $Q'$ as follows:
\begin{equation}
    \label{decS}
    S = S^0 +\sum_{k} S_{\mu_1\,\mu_2\,\ldots\,\mu_k} e^{\mu_1\,\mu_2\,\ldots\,\mu_k}, 
\end{equation}
    where $\mu_i=4,6,8,\ldots$ and $S^{0}\in Q'^{(0)}$. Additionally, if $k=2s$, then $S_{\mu_1\,\mu_2\,\ldots\,\mu_k}\in Q'^{(0)}$, otherwise $S_{\mu_1\,\mu_2\,\ldots\,\mu_k}\in Q'^{(1)}$. 

    \begin{ex}
    Let us consider the element $S=e+e^{12}+e^{45}+e^{1234}+e^{012345} \in \cl^{(0)}_{1,5}$. Then, the relevant subalgebra has the form $Q' =\cl(e^0,e^1,e^2,e^3,e^5)$. We can decompose $S$ as follows:
    $$S = S^0 + S_4 e^4,$$
    where
    $$S^0 = e+e^{12} \in Q'^{(0)},\quad S_{4}=-e^5+e^{123} -e^{01235}\in Q'^{(1)}.$$
\end{ex}
    
    \begin{prop}
\label{SPsi}
    Let the Hermitian idempotent $t$ have form \eqref{event} (\eqref{oddtsemi} or \eqref{oddtdouble}) and $\psi \in L(t)$. Suppose $Q'$ is defined by \eqref{Q1} (\eqref{Q2} or \eqref{Q3} respectively), then there uniquely exists $\Psi\in Q'^{(0)}$ such that $\psi = \Psi t$. 
    
    If $S\in \cl^{(0)}_{1,n}$ and $\widehat\psi=S \psi$, then $\widehat\Psi$ corresponded $\widehat\psi$ has the form:
    $$\widehat\Psi = S^0 \Psi +\sum_k S_{\mu_1\,\mu_2\,\ldots\,\mu_k} \Psi e^{\mu_1-1\,\mu_2-1\,\ldots\,\mu_k-1} I^k \in Q'^{(0)}.$$

    If $S\in \cl^{(1)}_{1,n}$ and $\widehat\psi=S \psi$, then $\widehat\Psi$ corresponded $\widehat\psi$ has the form:
    $$\widehat\Psi = \left(S^0 \Psi +\sum_k S_{\mu_1\,\mu_2\,\ldots\,\mu_k} \Psi e^{\mu_1-1\,\mu_2-1\,\ldots\,\mu_k-1} I^k\right)E \in Q'^{(0)}.$$
    
\end{prop}

\begin{proof}
We begin by considering the case when $S\in \cl_{1,n}^{(0)}$. The second case follows similarly, taking into account property \eqref{ItEt}. Using decomposition~\eqref{decS}, we express $S\Psi t$ as:
      $$S\Psi t = \left( S^0 +\sum_k S_{\mu_1\,\mu_2\,\ldots\,\mu_k} e^{\mu_1\,\mu_2\,\ldots\,\mu_k}\right) \Psi t.$$

 Next, we examine the summation terms. Since $\Psi\in Q'^{(0)}$ is an even element that does not contain the basis elements $e^{\mu_1\,\mu_2\,\ldots\,\mu_k}$, it commutes with them, allowing us to rewrite the expression as:
 $$S\Psi t = \left( S^0 \Psi +\sum_k S_{\mu_1\,\mu_2\,\ldots\,\mu_k} \Psi e^{\mu_1\,\mu_2\,\ldots\,\mu_k}\right)  t.$$

To simplify the terms further, we use property \eqref{propt} in the following form:
    $$e^{2\mu} t = e^{2\mu-1} I t.$$
Applying this relation, the product of a basis element with the idempotent transforms as:
$$e^{\mu_1\,\mu_2\,\ldots\,\mu_k} t = e^{\mu_1-1\,\mu_2-1\,\ldots\,\mu_k-1} I^k t.$$

Substituting this into the previous expression, we obtain:
  $$S\Psi t = \left( S^0 \Psi +\sum_k S_{\mu_1\,\mu_2\,\ldots\,\mu_k} \Psi e^{\mu_1-1\,\mu_2-1\,\ldots\,\mu_k-1} I^k\right)  t.$$
\end{proof}

\begin{thm}
     Let us fix the primitive Hermitian idempotent $t$ as in formula (\ref{event}) and the real algebra $Q'$ be defined by \eqref{Q1}.  Consider the following transformations:
 \begin{align*}
    &x^{\mu} \mapsto \widehat{x^{\mu}} = p_{\nu}^{\mu} x^{\nu}, \quad P=(p_{\nu}^{\mu})\in 
        {\rm O}(1,2d-1), \\
   & \partial_{\mu} \mapsto \widehat\partial_{\mu} = q_{\mu}^{\nu} \partial_{\nu}, \quad Q=(q_{\mu}^{\nu})=P^{-1}, \\
   & a_{\mu} \mapsto \widehat{a}_{\mu} = q_{\mu}^{\nu} a_{\nu}, \\
    &e^{\mu} \mapsto \widehat{e^{\mu}}=e^{\mu},
\end{align*}
where $\mu,\,\nu=0,1,2,\ldots,2d-1$.
Also, assume that the element $S$ is related to the matrix $P$ as in \eqref{SP}.
     
     If $\Psi\in Q'^{(0)}$ is a solution of equation \eqref{Dir-Hest} and $S\in {\rm Spin}(1,2d-1)$, then the function $\widehat\Psi$
   $$\widehat\Psi = S^0 \Psi +\sum_k S_{\mu_1\,\mu_2\,\ldots\,\mu_k} \Psi e^{\mu_1-1\,\mu_2-1\,\ldots\,\mu_k-1} I^k$$
 is a solution of transformed Dirac--Hestenes equation \eqref{TransfDir-Hest}.

 If $\Psi\in Q'^{(0)}$ is a solution of equation \eqref{Dir-Hest} and $S\in  {\rm Pin}(1,2d-1)\setminus {\rm Spin}(1,2d-1)$, then the function $\widehat\Psi$
   $$\widehat\Psi = S^0 \Psi E +\sum_k S_{\mu_1\,\mu_2\,\ldots\,\mu_k} \Psi e^{\mu_1-1\,\mu_2-1\,\ldots\,\mu_k-1} I^k E$$
 is a solution of  transformed Dirac--Hestenes equation \eqref{TransfDir-Hest}.
\end{thm}

\begin{proof}
To establish the invariance of the Dirac--Hestenes equation, we substitute the transformed variables into the left-hand side of equation \eqref{Dir-Hest} and verify that the resulting expression equals zero. Due to Proposition \ref{SPsi}, the transformed wave function satisfies the relation:
$$\widehat \Psi t = S \Psi t.$$

Let us denote the left-hand side of  transformed Dirac--Hestenes equation~\eqref{TransfDir-Hest} by $\widehat{F}$.
The variable $\widehat{A}_{\mu}$ satisfies the property:
$$\widehat{A}_{\mu} t = S q_{\mu}^{\nu} A_{\nu}t.$$

The generators $\{e^{\mu}\}$ remain unchanged under Lorentz transformation and $\widehat A\in Q'^{(0)}$. Therefore,  we multiply $\widehat F$ by the idempotent $t$ and, using Proposition \ref{sum2}, merge the two sums:
$$\widehat F t =  e^{\mu} \widehat A_{\mu} E t +m\widehat\Psi I t.$$

Substituting the definitions of the transformed variables, we obtain:
$$\widehat F t =  e^{\mu} S q_{\mu}^{\nu}A_{\nu} E t +m S\Psi I t.$$

Putting $S$ and $t$ out of brackets, we rewrite this in the following form:
$$\widehat F t = S\left( S^{-1} e^{\mu} S  q_{\mu}^{\nu}A_{\nu} E +m\Psi I\right) t.$$

Using property \eqref{SP} and the fact that $Q$ is inverse of $P$, we simplify further:
$$\widehat F t = S\left( e^{\nu} A_{\nu} E +m\Psi I\right) t.$$

We can separate the sum into two sums due to Proposition \ref{sum2}. The wave function $\Psi$ is a solution to equation \eqref{Dir-Hest}, that is why:
$$\widehat F t =0.$$

Applying Lemma \ref{Yt=0} to $\widehat F\in Q'^{(0)}$, we conclude that $\widehat F=0$.
\end{proof}

\begin{thm}
     Let us fix the primitive Hermitian idempotent $t$ as in formula \eqref{oddtsemi} (or \eqref{oddtdouble}) and the real algebra $Q'$ be  defined by \eqref{Q2} (or \eqref{Q3}  respectively). Consider the following transformations:
 \begin{align*}
    &x^{\mu} \mapsto \widehat{x^{\mu}} = p_{\nu}^{\mu} x^{\nu}, \quad P=(p_{\nu}^{\mu})\in 
        {\rm SO}(1,2d), \\
   & \partial_{\mu} \mapsto \widehat\partial_{\mu} = q_{\mu}^{\nu} \partial_{\nu}, \quad Q=(q_{\mu}^{\nu})=P^{-1}, \\
   & a_{\mu} \mapsto \widehat{a}_{\mu} = q_{\mu}^{\nu} a_{\nu}, \\
    &e^{\mu} \mapsto \widehat{e^{\mu}}=e^{\mu},
\end{align*}
where $\mu,\,\nu=0,1,2,\ldots,2d$. Also, assume that the element $S$ is related to the matrix $P$ as in \eqref{SP}.
     
     If $\Psi\in Q'^{(0)}$ is a solution of equation \eqref{SemispinorDir-Hest} (or \eqref{DoublespinorDir-Hest} respectively) and $S\in {\rm Spin}(1,2d)$, then the function $\widehat\Psi$
   $$\widehat\Psi = S^0 \Psi +\sum_k S_{\mu_1\,\mu_2\,\ldots\,\mu_k} \Psi e^{\mu_1-1\,\mu_2-1\,\ldots\,\mu_k-1} I^k$$
 is a solution of transformed Dirac--Hestenes equation \eqref{TransfSemispinorDir-Hest} (or \eqref{TransfDoublespinorDir-Hest} respectively).

 \end{thm}

 In this section, we have demonstrated the Lorentz invariance of the multidimensional Dirac--Hestenes equation using two approaches: the tensor approach and the spinor approach. The tensor approach requires additional invariants, whereas the spinor approach eliminates this necessity. However, in the spinor approach, the element $S$ of the spin group must be modified as described in Proposition \ref{SPsi}. This modification ensures compatibility with the structure of the algebra $Q'^{(0)}$, to which solutions of the multidimensional Dirac--Hestenes equation belong.

\section{Conclusion}
\label{sec:Con}

Lorentz transformations are defined as the orthogonal transformations preserving the quadratic form of signature $(1,n)$. Due to the homomorphism between the pseudo-orthogonal and spin groups, transformations can be formulated in two distinct ways for the Dirac--Hestenes equation: the tensor and the spinor approaches. In the tensor approach, the generators of the geometric algebra transform under the pseudo-orthogonal group. In the spinor approach, the wave function itself transforms under the spin group. However, the specific realization of these transformations depends on the parity of $n$, owing to the structure of the kernel of the homomorphism between the pseudo-orthogonal and spin groups. For the case $n=2d$, we have restricted our analysis to the proper subgroup ${\rm Spin}(1,2d)$ rather than the full ${\rm Pin}(1,2d)$ group to ensure a consistent and well-defined transformation law. 

In this paper, we have investigated the properties of the multidimensional Dirac--Hestenes equation in the real subalgebra $Q'^{(0)}$ of the geometric algebra $\cl_{1,n}$. We have explored both tensor and spinor approaches to analyze the behavior of the equation under Lorentz transformations, showing that the equation remains invariant in both formalisms. The tensor approach requires additional invariants, while the spinor approach, although more algebraically involved, naturally preserves the structure of the wave function. Moreover, we have derived an explicit expression for the transformed wave function in the spinor approach. This result is nontrivial, as the wave function belongs to the algebra $Q'^{(0)}$, while the spin group elements do not.

\subsection*{Acknowledgment}

This work is supported by the Russian Science Foundation (project 23-71-10028), https://rscf.ru/en/project/23-71-10028/.

\medskip

\noindent{\bf Data availability} Data sharing is not applicable to this article as no datasets were generated or analyzed during the current study.

\medskip

\noindent{\bf Declarations}\\
\noindent{\bf Conflict of interest} The authors declare that they have no conflict of interest.

\bibliographystyle{spmpsci}
\bibliography{myBibLib} 

@preamble{ "\newcommand{\noopsort}[1]{} "
        # "\newcommand{\printfirst}[2]{#1} "
        # "\newcommand{\singleletter}[1]{#1} "
        # "\newcommand{\switchargs}[2]{#2#1} " }

@incollection{lounesto,
author="Lounesto, Pertti",
editor="Chisholm, J. S. R.
and Common, A. K.",
title="Clifford {A}lgebras and {S}pinors",
bookTitle="Clifford Algebras and Their Applications in Mathematical Physics",
year="1986",
publisher="Springer Netherlands",
address="Dordrecht",
pages="25--37",
}

@article{Hestenes,
  title={Real spinor fields},
  author={Hestenes, David},
  journal={J. Math. Phys.},
  volume={8},
  number={4},
  pages={798--808},
  year={1967},
  publisher={AIP Publishing}
}

@book{Benn,
  title={An introduction to spinors and geometry with applications in physics},
  author={Benn, Ian M and Tucker, Robin W},
  publisher={Adam Hilger},
  year={1987}
}

@article{ndimension,
  title={Pauli theorem in the description of $n$-dimensional spinors in the {C}lifford algebra formalism},
  author={Shirokov, Dmitry Sergeevich},
  journal={Theor. Math. Phys.},
  volume={175},
  pages={454--474},
  year={2013},
  publisher={Springer}
}

@article{Traubenberg,
  title={Clifford algebras in physics},
  author={de Traubenberg, Michel Rausch},
  journal={AACA},
  volume={19},
  pages={869--908},
  year={2009},
  publisher={Springer}
}

@book{hestenes2012clifford,
  title={Clifford {A}lgebra to {G}eometric {C}alculus: {A} unified language for mathematics and physics},
  author={Hestenes, David and Sobczyk, Garret},
  year={1984},
  publisher={D. Reidel}
}

@article{Hestenes2,
  title={Spacetime physics with geometric algebra},
  author={Hestenes, David},
  journal={Am. J. Phys.},
  volume={71},
  number={7},
  pages={691--714},
  year={2003},
  publisher={American Association of Physics Teachers}
}

@book{Lasenby,
  title={Geometric algebra for physicists},
  author={Doran, Chris and Lasenby, Anthony},
  year={2003},
  publisher={Cambridge University Press}
}

@article{Shirokov,
  title={Unitary spaces on {C}lifford algebras},
  author={Marchuk, N. G. and Shirokov, D. S.},
  journal={AACA},
  volume={18},
  pages={237--254},
  year={2008},
  publisher={Springer}
}

@book{Sommerfeld,
title={Atombau und Spektrallinien, 2 Band},
author={Sommerfeld, A},
publisher={Friedrich Vieweg \& Sohn},
year={1951}

}

@article{Riesz,
    author = {Riesz, M.},
    title = {Sur {C}ertaines {N}otions {F}ondamentales en {T}héorie {Q}uantique {R}elativiste},
    journal = {Congres Math. Scandinaves},
    pages={123--148},
    year = {1947}
}

@article{RumShir,
 author = {Rumyantseva, Sofia and Shirokov, Dmitry},
    title = {Introducing multidimensional {D}irac--{H}estenes Equation},
 journal = {AACA},
volume={35},
pages={24},
    year = {2025}
}

@article{BrauerWeyl,
    author ={Brauer, R. and Weyl, H.},
    title ={Spinors in n dimensions},
    journal = {American Journal of Mathematics},
    year = {1935},
    volume={57},
    pages={425--449}
}

@book{kaku1993quantum,
  title={Quantum field theory: a modern introduction},
  author={Kaku, Michio},
  year={1993},
  publisher={Oxford university press}
}

@article{hestenes1990zitterbewegung,
  title={The zitterbewegung interpretation of quantum mechanics},
  author={Hestenes, David},
  journal={Foundations of Physics},
  volume={20},
  number={10},
  pages={1213--1232},
  year={1990},
  publisher={Springer}
}

@article{andoni2024spin,
  title={Spin 1/2 one-and two-particle systems in physical space without eigen-algebra or tensor product},
  author={Andoni, Sokol},
  journal={Mathematical Methods in the Applied Sciences},
  volume={47},
  number={3},
  pages={1457--1470},
  year={2024},
  publisher={Wiley Online Library}
}

@article{shirokov2013,
  title={The use of the generalized {P}auli's theorem for odd elements of {C}lifford algebrato analyze relations between spin and orthogonal groups of arbitrary dimensions},
  author={Shirokov, Dmitry Sergeevich},
  journal={Vestn. Samar. Gos. Tekhn. Univ., Ser. Fiz.-Mat. Nauki},
  volume={1},
  number={30},
  pages={279--287},
  year={2013}
}

@article{borvstnik2002generate,
  title={How to generate spinor representations in any dimension in terms of projection and nilpotent operators},
  author={Bor{\v{s}}tnik, N Manko{\v{c}} and Nielsen, Holger B},
  journal={Journal of Mathematical Physics},
  volume={43},
  number={11},
  pages={5782--5803},
  year={2002},
  publisher={American Institute of Physics}
}
\end{document}